\documentclass[11pt]{article}
 \usepackage[margin=1in]{geometry}
\usepackage{float,url}
\usepackage{amssymb,amsfonts,amsmath,amsthm}

\newtheorem{theorem}{Theorem}[section]
\newtheorem{lemma}{Lemma}[section]
\newtheorem{proposition}{Proposition}[section]
\newtheorem{corollary}{Corollary}[section]

\newtheorem{claim}{Claim}

\newcommand{\lv}[1]{}
\newcommand{\COC}{{\sc Component Order Connectivity}}
\newcommand{\DCOC}{{\sc Directed Component Order Connecti\-vi\-ty}}
\newcommand{\COCell}{{\sc \COC[$\ell$]}}
\newcommand{\COCk}{{\sc \COC[$k$]}}
\newcommand{\COCellk}{{\sc \COC[$\ell+k$]}}
\newcommand{\COCnell}{{\sc \COC[$n-\ell$]}}

\newcommand{\DCOCell}{{\sc \DCOC[$\ell$]}}
\newcommand{\DCOCk}{{\sc \DCOC[$k$]}}
\newcommand{\DCOCellk}{{\sc \DCOC[$\ell+k$]}}
\newcommand{\DCOCnell}{{\sc \DCOC[$n-\ell$]}}
\newcommand{\mco}{{\rm mco}}

\newcommand{\cF}{{\mathcal F}}

\newcommand{\bigoh}{{\mathcal O}}

\newcommand{\jbj}[1]{{#1}}
\newcommand{\gzg}[1]{{#1}}
\newcommand{\mw}[1]{{#1}}
\newcommand{\ee}[1]{{#1}}

\newcommand{\viol}[1]{{#1}}

\newcommand{\gutin}[1]{{#1}}

\newcommand{\eiben}[1]{{#1}}

\usepackage[obeyFinal,colorinlistoftodos,textsize=tiny]{todonotes}

\newcommand{\defproblem}[4]{
  \vspace{3mm}
\noindent\fbox{
  \begin{minipage}{.95\textwidth}
  \begin{tabular*}{\textwidth}{@{\extracolsep{\fill}}lr} #1  & #2 \\ \end{tabular*}
  {\bf{Input:}} #3  \\
  {\bf{Question:}} #4
  \end{minipage}
  }
  \vspace{2mm}
}

\title{Component Order Connectivity in Directed Graphs} 

\author{
J\o{}rgen Bang-Jensen\thanks{University of Southern Denmark, Denmark.  {\tt jbj@imada.sdu.dk}. Research supported by the Independent Research Fund Denmark under grant number DFF 7014-00037B} \and Eduard Eiben\thanks{Royal Holloway, University of London, UK.  {\tt eduard.eiben@rhul.ac.uk}.} \and Gregory Gutin \thanks{Royal Holloway University of London, UK. {\tt g.gutin@rhul.ac.uk}.} \and Magnus Wahlstr{\"o}m\thanks{Royal Holloway, University of London, UK.  {\tt Magnus.Wahlstrom@rhul.ac.uk}.  }\and Anders Yeo\thanks{University of Southern Denmark, Denmark.  {\tt andersyeo@gmail.com}. Research supported by the Independent Research Fund Denmark under grant number DFF 7014-00037B} }
\begin{document}
 \maketitle

\begin{abstract}
A directed graph $D$ is semicomplete if for every pair $x,y$ of vertices of $D,$ there is at least one arc between $x$ and $y.$ \viol{Thus, a tournament is a semicomplete digraph.}
In the Directed Component Order Connectivity (DCOC) problem, given a digraph $D=(V,A)$ and a pair of natural numbers $k$ and $\ell$, we are to decide whether there is a subset $X$ of $V$ of size $k$ such that  the largest strong connectivity component in $D-X$ has at most $\ell$ vertices. Note that DCOC reduces to the Directed Feedback Vertex Set problem for $\ell=1.$
We study parametered complexity of DCOC for general and semicomplete digraphs with the following parameters: $k, \ell,\ell+k$ and $n-\ell$. 
In particular, we prove that DCOC with parameter $k$ on semicomplete digraphs can be solved in time $O^*(2^{16k})$ but not in time $O^*(2^{o(k)})$ unless the Exponential Time Hypothesis (ETH) fails.
\gutin{The upper bound $O^*(2^{16k})$ implies the upper bound $O^*(2^{16(n-\ell)})$ for the parameter $n-\ell.$ We complement the latter by showing that there is no algorithm
of  time complexity $O^*(2^{o({n-\ell})})$ unless ETH fails.}
Finally, we improve \viol{(in dependency on $\ell$)} the upper bound of G{\"{o}}ke, Marx and Mnich (2019) for the time complexity of DCOC with parameter $\ell+k$
on general digraphs from $O^*(2^{O(k\ell\log (k\ell))})$ to $O^*(2^{O(k\log (k\ell))}).$ Note that Drange, Dregi and van 't Hof (2016) proved that even for the undirected version of DCOC on split graphs
there is no algorithm of running time $O^*(2^{o(k\log \ell)})$ unless ETH fails and
it is a long-standing problem to decide whether Directed Feedback Vertex Set admits an algorithm of time complexity $O^*(2^{o(k\log k)}).$
\end{abstract}

\section{Introduction}

Motivated by various practical network applications, many different vulnerability measures of undirected graphs have
been introduced and studied in the literature. The two most studied of such measures are 
vertex and edge connectivity of an undirected graph. However, these two measures often do not capture the
more subtle vulnerability properties of networks that one might wish to consider, such as 
 the number of vertices in the largest remaining connected component.
 
 While both undirected and directed graphs are of great interest in graph theory and algorithms applications, undirected graphs have been studied
much more than their directed counterparts arguably due to simpler structure of undirected graphs. 
In this paper, we study a number of parameterizations of a problem of interest from
both theory and applications which was mainly studied for undirected graphs so far. 

In many networks, the underlying graph is directed rather than undirected and the aim of this paper is to study an 
extension to directed graphs of the {\bf $\ell$-component order connectivity} of an undirected graph $G$, which is the size of a minimum set $X\subseteq V(G)$  
such that $\mco(G-X)\le \ell,$ where $\mco(G-X)$ is the number of vertices in the largest connected component of $G-X$  (mco stands for
maximum component order). By \COC{} will denote the following decision problem:

\defproblem{{\sc component order connectivity}}{{}}{A graph $G=(V,E)$ and a pair $\ell,k\in{\mathbb N}$ of natural numbers}{
  Is there  a subset $X$ of $V$ of size $k$ such that $\mco(G-X)\le \ell.$}

For a survey on \COC, see Gross et al. \cite{gross13}; for more recent research on the problem, see e.g.
\cite{DrangeDH16,KumarL16,Lee19}. 
 
For a directed graph $D,$  we define the {\bf $\ell$-component order connectivity} as the size of a minimum set $X\subseteq V(D)$  
such that $\mco(D-X)\le \ell,$ where $\mco(D-X)$ is the number of vertices in the largest strongly connected component of $D-X.$ 
Using this definition of $\mco(D-X),$ we state can state the following directed version of  \COC{}.

\defproblem{{\sc directed component order connectivity}}{{}}{A digraph $D=(V,A)$ and a pair $\ell,k\in{\mathbb N}$ of natural numbers}{
  Is there  a subset $X$ of $V$ of size $k$ such that $\mco(D-X)\le \ell.$}
  
 \gzg{In what follows, we will assume without loss of generality that $k+\ell< n=|V|$ (or, $k< n-\ell$). Indeed, if $k+\ell\ge n$ then our instance is a YES-instance since deleting any set $X$ of $k$ vertices
implies  $\mco(D-X)\le \ell.$
}

Clearly, \DCOC{} is a generalization of \COC{} (each instance $(G,\ell,k)$ of \COC{} corresponds to an equivalent instance $(D,\ell,k)$ of \DCOC, where $D$ is obtained from 
$G$ by replacing every edge of $G$ by a directed 2-cycle). For $\ell=1,$ while \COC{} is equivalent to the {\sc Vertex Cover} problem,
 \DCOC{} is equivalent to the {\sc Directed Feedback Vertex Set} problem. Unlike  {\sc Vertex Cover} whose fixed-parameter tractability is very easy to show, \jbj{a fact 
that was known} very early on in parameterized algorithmics \cite{DowneyF99}, fixed-parameter tractability of {\sc Directed Feedback Vertex Set} 
 was a long-standing open problem until Chen et al. \cite{chenJACM55} in 2008 proved its fixed-parameter tractability by designing a $4^kk!n^{\bigoh(1)}$-time algorithm.
 We provide basics on parameterized algorithms and complexity in the next section.

Since \COC{} is NP-complete (it remains NP-complete even for split, co-bipartite and chordal undirected graphs \cite{DrangeDH16}), a number of researchers
studied \COC{} using the framework of parameterized algorithmics, see e.g. \cite{DrangeDH16,KumarL16,Lee19}.
G\"oke, Marx and Mnich~\cite{GokeMM20arXiv} were the first to study the \DCOC{}  problem  \jbj{from the viewpoint of } parameterized algorithms and complexity. 
They obtained an algorithm of 
running time $4^k(k\ell+k+\ell)!n^{\bigoh(1)},$ which is close to the complexity of the algorithm of Chen et al. \cite{chenJACM55} when $\ell=1$.
Thus,  \DCOC{} parameterized by $k+\ell$ is fixed-parameter tractable (FPT).

We will continue the study of \DCOC{} using parameterized algorithms and complexity. 
In particular, as in papers \cite{DrangeDH16,KumarL16,Lee19} which studied \COC,
we study  \DCOC{} parameterized by three parameters: $\ell$, $k$ and $\ell+k.$ We will denote the corresponding
parameterized problems by \DCOCell, \DCOCk{} and \DCOCellk, respectively.

Moreover, we introduce and study a new parameterization of \DCOC: parameter $n-\ell,$ where $n$ is the number of vertices in $D.$ 
One reason to introduce \DCOCnell{} is that normally one requires the parameters to be relatively small compaired to the size of the problem under 
consideration. However, if $k$ is small it is possible that for every $X\subseteq V(D)$  of size $k$, $\mco(D-X)$ is not much smaller than $n-k.$ 
Then  $n-\ell$ can be much smaller than $\ell.$ 

Since  \COC{} is equivalent to the {\sc Vertex Cover} problem for $\ell=1$, \COCell{} is para-NP-complete. 
\gzg{Drange et al. \cite[Theorem 8]{DrangeDH16} proved that \COCk{} is W[1]-hard even on split graphs. In their construction, 
$n-\ell=\bigoh(k^2).$} Hence, \COCnell{}  is also W[1]-hard. 
They also showed that \COCellk{} is FPT by obtaining an algorithm of running time $2^{\bigoh(k\log \ell)}n.$ 
The above mentioned results are written in the undirected graphs row of Table \ref{fig:res}. 

A directed graph $D$ is {\bf semicomplete} if for every pair $x,y$ of distinct vertices of $D$, there is an arc between $x$ and $y.$ 
When we require that there is only one arc between $x$ and $y$ then we obtain a definition of a {\bf tournament}. 
Clearly, the hardness results for the directed graphs row of Table \ref{fig:res} follow from the corresponding results in the undirected graphs row for columns $n-\ell$ and $k$.
\DCOCell{} is para-NP-complete for semicomplete digraphs  as \DCOC{} on semicomplete digraphs
is NP-complete for $\ell=1$. \jbj{This follows from the fact that} {\sc Directed Feedback Vertex Set} is NP-complete
even for tournaments, as proved by Bang-Jensen and Thomassen \cite{Bang-JensenT92} and Speckenmeyer \cite{Speckenmeyer89}. 

The FPT result in directed graphs row of Table \ref{fig:res} is first obtained by G\"oke et al.~\cite{GokeMM20arXiv} as discussed above. The running time of their algorithm
is $
  4^k(k\ell+k+\ell)!n^{\bigoh(1)} = 2^{\bigoh(k \ell \log(k \ell))} n^{\bigoh(1)}.
$
By modifying their algorithm, we obtained an algorithm of complexity  $2^{\bigoh(k)} \ell^k  k! n^{\bigoh(1)}=2^{\bigoh(k \log(k \ell))} n^{\bigoh(1)},$ which decreases asymptotic 
dependence of the running time on $\ell.$\footnote{We obtained this result independently from \cite{Neogi+}; our approach is different from that in \cite{Neogi+}.}
Our modification consists of replacing a branching algorithm in \cite{GokeMM20arXiv} with a randomized algorithm which can be derandomized without increasing the complexity upper bound. 
\gutin{Note that Drange et al. \cite[Theorem 14]{DrangeDH16} proved that even for \COC{} on split graphs
there is no algorithm of running time $O^*(2^{o(k\log \ell)})$ (\viol{here we assume that $\ell=O(k^{O(1)})$}) unless the Exponential Time Hypothesis (ETH)~\cite{ImpagliazzoRF2001} fails and
it is a long-standing problem to decide whether {\sc Directed Feedback Vertex Set} admits an algorithm of time complexity $O^*(2^{o(k\log k)}).$}

\begin{table}[H]\label{tab}
\caption{Parameterized Complexity of ({\sc Directed}) \COC}
\begin{center}
\begin{tabular}{|c|c|c|c|c|} \hline
class of graphs & $n-\ell$ & $k$ & $\ell$ & $\ell+k$\\ \hline
semicomplete digraphs & FPT & FPT & para-NP-c. & FPT\\
undirected graphs & W[1]-hard&W[1]-hard& para-NP-c. & FPT\\
directed graphs & W[1]-hard&W[1]-hard& para-NP-c. & FPT\\
\hline 
\end{tabular}
\end{center}
\label{fig:res}
\end{table}%

The most interesting entry in the semicomplete digraphs row is a non-trivial result that \DCOCk{} on semicomplete digraphs
is FPT. This FPT algorithm boils down to finding a shortest path in a suitably defined  auxiliary weighted acyclic digraph.
The running time of the algorithm is $\bigoh(2^{16k}kn^2).$ 
The other two FPT entries in this row follow from this result (for the parameter $n-\ell$ this is due to our assumption that $k< n-\ell$).
We also prove the following lower bounds: no algorithm for  \DCOCk{} on semicomplete digraphs
can have time complexity $2^{o(k)}n^{\bigoh(1)}$ unless ETH fails\footnote{\viol{Similarly, no algorithm for \DCOCnell{} on semicomplete digraphs can have running time $2^{o(n-\ell)} n^{\bigoh(1)}$, unless ETH fails. }} and no such deterministic algorithm can run in time $o(n^2).$  


Our paper is organised as follows. The next section is devoted to terminology and notation on directed and undirected graphs, and basics on parameterized algorithms and complexity. 
In Section \ref{sec:gd}, we describe our improvement on the algorithm of G\"oke et al.~\cite{GokeMM20arXiv}. 
In Section \ref{sec:sd}, we prove that \DCOCk{} on semicomplete digraphs
admits an algorithm of running time $\bigoh^*(2^{16k})$ \gzg{and show the lower bounds on running time with parameters $k$ and $n-\ell$.}
We conclude the paper in Section \ref{sec:c}.


\section{Preliminaries}

\subsection{Directed and Undirected Graph Terminology and Notation}\label{ssec:pre}

In this paper, all directed and undirected graphs are finite, without loops or parallel edges.
As often the case in the directed graph theory, an edge of a digraph will be called an {\bf arc} 
and the vertex and arc sets of a digraph $D$ will be denoted by $V(D)$ and $A(D),$ respectively.
The {\bf out-neighbourhood} and {\bf in-neighbourhood} of a vertex $x$ of a digraph $D$
are denoted by $N^+_D(x)=\{y\in V(D):\ xy\in A(D)\}$ and $N^-_D(x)=\{y\in V(D):\ yx\in A(D)\}$, respectively, and the subscript $D$ will be 
omitted if $D$ is clear from the context. The {\bf out-degree} and {\bf in-degree}
of a vertex $x$ of $D$ is $d^+_D(x)=|N^+_D(x)|$ and $d^-_D(x)=|N^-_D(x)|,$ respectively.

In this paper all paths and cycles in digraphs  are directed, so we will omit
the adjective `directed' when referring to paths and cycles in digraphs.
\jbj{If $D=(V,A)$ is a digraph and $S\subseteq V$, then we denote by $D[S]$ the subdigraph induced by the vertices in $S$.}
A digraph $D$ is {\bf strongly connected} (or, just {\bf strong}) if there is a path 
from $x$ to $y$ for every ordered pair $x,y$ of distinct vertices.
A {\bf strong component} of a digraph $D$ is a maximal strong induced subgraph of $D.$ 
Strong components of $D$ do not share vertices and
can be ordered $D_1,D_2,\dots ,D_p$ such that there is no arc \jbj{in $D$ from $V(D_j)$ to $V(D_i)$ when} $j>i.$ 
Such an ordering is called an {\bf acyclic ordering}. 
Note that if $D$ is a semicomplete digraph, then \jbj{the strong components of $D$ have} a unique acyclic ordering $D_1,D_2,\dots ,D_p$ \jbj{and we have} $xy\in A(D)$ 
for every $x\in V(D_i),\ y\in V(D_j),\ i<j.$

Basic digraph terminology not introduced in this section can be found in
\cite{bang2009,bang2018}.

\subsection{Parameterized Complexity}

An instance of a parameterized problem $\Pi$
is a pair $(I,k)$ where $I$ is the {\bf main part} and $k$ is the
{\bf parameter}; the latter is usually a non-negative integer.  
A parameterized problem is
{\bf fixed-parameter tractable} (FPT) if there exists a computable function
$f$ such that instances $(I,k)$ can be solved in time $\bigoh(f(k)|{I}|^c)$
where $|I|$ denotes the size of~$I$ and $c$ is an absolute constant. The class of all fixed-parameter
tractable decision problems is called {FPT} and algorithms which run in
the time specified above are called {FPT} algorithms. As in other literature on {FPT} algorithms,
we will sometimes omit the polynomial factor in $\bigoh(f(k)|{I}|^c)$ and write $\bigoh^*(f(k))$ instead.

\gutin{
While FPT is a parameterized complexity analog of {\sf P} in classic complexity, there are many hardness classes in parameterized 
complexity and they form a nested sequence starting from W[1]. It is well known that if the Exponential Time Hypothesis holds then FPT$\ne$W[1].
Due to this and other complexity results,  it is widely believed that FPT$\ne$W[1] and hence W[1] is viewed as a parameterized analog of {\sf NP} in classical complexity.
}


{\bf para-NP} is the class of parameterized problems which can be solved by a nondeterministic algorithm in time 
$\bigoh(f(k)|{I}|^c),$ where $f$ is a computable function and $c$ is an absolute constant.
It is well-known that if a problem $\Pi$ with parameter $\kappa$ is NP-complete when $\kappa$ equals
 to some constant, then $\Pi$ is para-NP-complete. It is also well known that FPT=para-NP if and only if P=NP.

For more information on parameterized algorithms and complexity, see recent books \cite{cygan2015,downey2013}.

\section{\DCOCellk{} on General Digraphs}\label{sec:gd}

G\"oke, Marx and Mnich~\cite{GokeMM20arXiv} showed that \DCOCellk{} is
FPT with a running time given as
\[
  4^k(k\ell+k+\ell)!n^{\bigoh(1)} = 2^{\bigoh(k \ell \log(k \ell))} n^{\bigoh(1)}.
\]
The core of their algorithm is as follows.  Begin with the
\emph{iterative compression} version of the problem, where in addition
to $(D, \ell, k)$ the input also contains a solution $X_0$ with $|X_0|=k+1$,
which can be used to guide the search for a smaller solution.
This is a standard ingredient in FPT algorithms; see, e.g.,~\cite{cygan2015}.
At the cost of a simple branching step, we may also assume that we are
looking for a solution $X$ with $X \cap X_0 = \emptyset$.
Next, they observe that \emph{if} we knew the strongly connected
components of $D-X$ that the vertices of $X_0$ are contained in,
then the problem reduces to a previously studied, simpler problem
known as \textsc{Skew Separator}~\cite{chenJACM55}, which occurs in the design of the
FPT algorithm for \textsc{Directed Feedback Vertex Set} (DFVS) of Chen
et al.~\cite{chenJACM55}.  Indeed, if the precise strong components containing the vertices of
$X_0$ are known, then the problem can be solved in time $O^*(4^kk!)$
using a strategy much like that for DFVS~\cite{chenJACM55,GokeMM20arXiv}.
Hence the bottleneck in \DCOCellk{} is the guessing of the strong
components of $X_0$ in $D-X$.

G\"oke et al.~\cite{GokeMM20arXiv} solve this via a branching
algorithm that they analyse as taking time at most $(k\ell + k + \ell)!$.
We show a simpler randomized method solving this problem with an
improved time bound of
\begin{equation}\label{eq:1}
\binom{\ell (k+1) + k}{k} \leq (e (\ell + 1+\ell/k))^k\eiben{\le (3e\ell)^k=2^{\bigoh(k)}\cdot \ell^k}
\end{equation}
The method can be derandomized by standard means.

\begin{lemma} \label{lemma:guess-components}
  Let $(D, \ell, k)$ be an instance of \DCOCellk, and let $X_0$ be a
  solution with $|X_0| = k+1$.  Let $X$ be an unknown solution with $|X| \leq k$
  \mw{such that $X \cap X_0 = \emptyset$}.
  There is a randomized procedure that with success probability at least 
  \[
    \left((\ell + k)^{\bigoh(1)} \binom{\ell k + \ell + k}{k} \right)^{-1}
  \]
  computes a set $S \subset V(D)$ such that for every $x \in X_0$,
  the strong components containing $x$ in $D-X$ and in $D[S]$ are identical.
\end{lemma}
\begin{proof}
  Initialize $S=X_0$, then for every vertex $v \in V(D) \setminus X_0$
  place $v$ in $S$ independently at random with probability $p=1-1/(\ell+1)$.
  We declare a guess a \emph{success} if the following conditions apply:
  \begin{enumerate}
  \item For every $x \in X_0$ we have $V_x \subseteq S$, 
    where $V_x \subseteq V$ is the strong component of $D-X$
    containing $x$
  \item $X \cap S = \emptyset$
  \end{enumerate}
 
  Let $Y=\bigcup_{x \in X_0} V_x$.  Our guess is successful if and
  only if $v \in S$ for every $v \in Y$, and $v \notin S$ for every $v \in X$. 
  Since these are independent events, this clearly happens with probability precisely
  \[
    p^{|Y|}(1-p)^{|X|} \geq p^{\ell(k+1)}(1-p)^k,
  \]
  hence the worst case occurs when all sets $V_x$ are disjoint and
  have $|V_x|=\ell$, and $|X|=k$, i.e., $|Y|=\ell(k+1)$ and $|X|=k$. 
  Let $S_0=X \cup Y$.  We bound the probability of success carefully
  in two steps:
  \begin{enumerate}
  \item We estimate the probability that $|S \cap S_0|=|Y|$,
    \emph{without} caring about the precise intersection (i.e.,
    success in this stage includes cases where $X \cap S \neq \emptyset$).
  \item We estimate the probability of success, conditional on the
    previous event. 
  \end{enumerate}
  Note $|S_0|=\ell(k+1)+k$ by assumption. 

  For the first step, note that the expected number of vertices of
  $S_0$ \emph{not} in $S$ is 
  \[
    (1-p)|S_0|=(1/(\ell+1))(\ell k+\ell+k) = k + \frac{\ell}{\ell+1}.
  \]
  Also note that in a successful guess, this value is precisely $k$.
  Hence the expected value differs from the intended value by less
  than 1. Since $|S \cap S_0|$ is a binomial distribution, due to the
  guesses being independent, this clearly happens with probability 
  at least inverse polynomial in $k+\ell$.

  Subject to this event, the set $S_0 \setminus S$ is uniformly
  distributed among all subsets of $S_0$ of size $k$ by independence,
  hence the conditional probability of success is one in 
  $\binom{\ell k + \ell + k}{k}$.  We conclude that the success
  probability matches the bound in the lemma.

  Finally, assume that the guess was successful for some set $S$
  and consider the strong component of $x$ in $D[S]$ for some $x \in X_0$.
  Let $V_x'$ be this strong component. Since $D[V_x]$ is strongly connected
  and $V_x \subseteq S$, we have $V_x \subseteq V_x'$.  On the other hand,
  by assumption $D[S]$ is an induced subgraph of $D-X$, and since
  $V_x$ is a strongly connected component in $D-X$ we must have $V_x' \subseteq V_x$.
  We conclude $V_x=V_x'$ for each $x \in X_0$, as required. 
\end{proof}

For the derandomization, we employ a {\bf cover-free family}
construction of Bshouty and Gabizon~\cite{BshoutyG17CIAC}.
We get the following:

\begin{lemma} \label{lemma:components-derand}
  There is a deterministic procedure that produces a set $\cF \subseteq 2^V$
  with
  \[
    |\cF| = \binom{\ell k + \ell  + k}{k}^{1+o(1)} \log |V|
  \]
  in time $\bigoh(|\cF| n)$, such that there is a set $S \in \cF$
  such that for every $x \in X_0$,
  the strong components containing $x$ in $D-X$ and in $D[S]$ are identical.
\end{lemma}
\begin{proof}
  Let $r \leq s < n$ be integers. 
  Bshouty and Gabizon (in a slightly non-standard definition)
  define an {\bf $(n,(r,s))$-cover free family} as a set $\cF \subseteq \{0,1\}^n$
  such that for every disjoint pair of sets $A, B \subseteq [n]$
  with $|A|=r$ and $|B|=s$ there is a set $S \in \cF$
  such that $A \subseteq S$ and $B \cap S = \emptyset$.
  Bshouty and Gabizon~\cite{BshoutyG17CIAC} show how to compute an
  $(n,(r,s))$-cover free family $\cF$ of size
  \[
    |\cF| = \binom{r+s}{r}^{1+o(1)} \log n
  \]
  in time $\bigoh(|\cF| n)$. 
  
  By Lemma~\ref{lemma:guess-components}, it suffices to construct a
  cover-free family with parameters $n=|V(D)|$, $r=\ell(k+1)$
  and $s=k$. Here $r>s$, but we can simply compute an
  $(n,(s,r))$-cover free family and take the complement of every member.
  Hence we get a family of size
  \[
    \binom{\ell k+\ell+k}{k}^{1+o(1)}\log n
  \]
  computed in output-linear time. 
\end{proof}

The two lemmas of this section \viol{and (\ref{eq:1})} imply the following:

\begin{theorem}
  There is a randomized FPT algorithm that solves \DCOCellk{} in time
  $2^{\bigoh(k)} \ell^k  k! n^{\bigoh(1)}$ with probability at least $\Omega(1)$.
  The algorithm can be derandomized in the same time, up to a
  lower-order overhead factor. 
\end{theorem}

\section{\DCOC{} on Semicomplete Digraphs}\label{sec:sd}
\eiben{Let us first summarize main ideas behind our FPT algorithm, before providing more technical details.}
 Let $D=(V,A)$ be a semicomplete digraph, $k,\ell\in \mathbb{N}$ and let $X\subseteq V$ of size $k$ such that $\mco(D-X)\le \ell$. The vertices of $D-X$ can be partitioned into  $C_1,\ldots, C_q$ such that \jbj{each $C_i$ is the vertex set of a strong component of $D-X$ and }
\begin{enumerate}
	\item for every $i\in [q]$ is  $|C_i|\le \ell$, and
	\item for every $i,j\in [q]$ with $i<j$ and every $x\in C_i$, $y\in C_j$ we have $xy\in A$ and $yx\notin A$. 
\end{enumerate}
\begin{figure}
	\centering
	\includegraphics{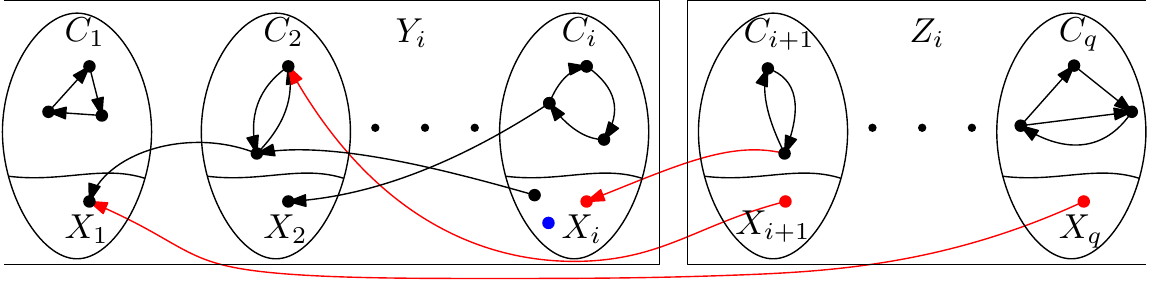}
	\caption{\eiben{An example of a valid triple $(Y_i, Z_i, S_i)$. A semicomplete digraph $D$, the set $X=\bigcup_{i\in [q]}X_i$ is such that $\mco(D-X)=3$ and $C_1,\ldots, C_q$ are strong components of $D-X$.  $Y_i=C'_1\cup C'_2\cup\cdots\cup C'_i$ and $Z_i=C'_{i+1}\cup C'_{i+2}\cup\cdots\cup C'_q$, where $C'_i=C_i\cup X_i$, $i\in [q]$. The arcs $uv$, $u\in C'_i$, $v\in C'_j$ for $i<j$ are omitted as well as the arcs between $X_i$ and $C_i$. The set $S_i$ is the set of the three red vertices, one in each of $X_i$, $X_{i+1}$, and $X_q$, is a minimal vertex cover of the red arcs from $Z_i$ to $Y_i$. Note that the vertex in $X_1$ is not in $S_i$ as the arc incident to it with the tail in $Z_i$ is already covered by $S_i$. Note also the blue vertex in $X_i$, the only reason it is in $X$ is to reduce the size of $C_i$ and as such it will not appear in any $S_j$, $j\in [q]$, in the set of $q$ valid triples defining these components.}}\label{fig:valid_triples}
\end{figure}
\eiben{In our algorithm, we would like to discover the strong components one by one in the ascending order from $C_1$ to $C_q$. Now let $X_1,\ldots, X_q$ be a partition of $X$ into $q$ (possibly empty) parts and let, for each $i\in [q]$,  $Y_i=C'_1\cup C'_2\cup\cdots\cup C'_i$ and $Z_i=C'_{i+1}\cup C'_{i+2}\cup\cdots\cup C'_q$, where $C'_i=C_i\cup X_i$, $i\in [q]$. Moreover, let $S_i$ be a subset of $X$ such that for each $y\in Y_i\setminus S_i$ and $z\in Z_i\setminus S_i$ we have $yz\in A$ and $zy\notin A$. See also Figure~\ref{fig:valid_triples}. Note that, given $S_i$, it suffice to solve our problem in subgraphs $D[Y_i\setminus S_i]$ and $D[Z_i\setminus S_i]$ separately. Moreover, the set $(Y_{i+1}\setminus Y_{i}) \setminus (S_{i+1}\cup S_i)$ is basically the strong component $C_{i+1}$ up to few vertices in $X_{i+1}$ that are not incident to any arc with tail in $Z_{i+1}\setminus S_{i+1}$ or head in $Y_i\setminus S_i$. Such vertices can actually be replaced in $X$ by any vertex in $C_{i+1}$. It follows, that if we are given $(Y_1,Z_1,S_1),\ldots , (Y_q,Z_q,S_q)$, then we can easily reconstruct a solution of size $|X|$ as $\bigcup_{i\in [q]}S_i$ plus some arbitrary vertices of $(Y_{i+1}\setminus Y_{i}) \setminus (S_{i+1}\cup S_i)$} 
\viol{to have at most  $\ell$ vertices in each strong component of $D-X$.}

 \eiben{Therefore, our goal will be to search for triples $(Y_i, Z_i, S_i)$, $i\in [q]$, where $\{Y_i,Z_i\}$ is a partition of $V$ and $S_i$ is a minimal subset of $X$ such that there is no arc $zy$ in $A$ with $z\in Z_i\setminus S_i$ and $y\in Y_i\setminus S_i$. 
	The first step of our proof is to show that there are at most $2^{8k+2}n$ triples we need to consider (Lemma~\ref{lemma:numberOfTripples}). We will call these important triples \textbf{valid} and we postpone the precise definition for later. The main reason for the bound is that we only need to consider triples $(Y_i, Z_i, S_i)$ for which $|S_i|\le k$ and that if we fix $|Y_i|$ (and hence also $|Z_i|$), then vertices with out-degree at least $|Z_i|+|S_i|+1$ (resp. in-degree at least $|Y_i|+|S_i|+1$) have to be in $Y_i$ (resp. in $Z_i$) or in $S_i$ and we can fix these vertices in $Y_i$ (resp. in $Z_i$). Once we bound the number of the triples we need to consider, we can define \textbf{compatible} pairs of triples $\left((Y^1, Z^1, S^1), (Y^2, Z^2, S^2)\right)$, for which $Y^1\subset Y^2$ and these triples, loosely speaking can define a strong component of $D-X$ with at most $\ell$ vertices as $(Y^2\setminus Y^1) \setminus (S^1\cup S^2)$ and the arcs from $Z_2$ to $Y_1$ are all hit by a vertex in $S^1\cap S^2$. This allows us to create an auxiliary acyclic ``state'' digraph whose vertices are valid triples and arcs are the compatible pairs of triples. The paths from $(\emptyset, V, \emptyset)$ to $(V, \emptyset, \emptyset)$ in this graph then define a solution for $(D,\ell, k)$. Note that our algorithm can be equivalently seen as a dynamic programming which computes for each valid triple $(Y,Z,S)$ a minimum size set $X$ such that $\mco(D[Y]-(X\cup S))\le \ell$.}

\subparagraph{}\eiben{
	The following lemma allows us to show that if we fix $|Y|$ in a triple $(Y,Z,S)$, then only $\bigoh(k)$ vertices of $D$ could potentially be in both $Y$ and $Z$ and all other vertices are fixed}.
\jbj{The 
	lemma is an easy consequence of the fact that every semicomplete digraph on at least $2p+2$, $p\in \mathbb{N}$, vertices has a vertex of out-degree at least $p+1$. We give the proof for the convenience of the reader.}

\begin{lemma}\label{lem:DegreeSeparation}
 Let $D=(V,A)$ be a semicomplete digraph and let $Y,Z$ be a partition of $V$ such that for every $y\in Y$ and every $z\in Z,$ we have $yz\in A$. Then for every $p\in \mathbb{N}$ (1) there are at most $2p+1$ vertices in $Y$ with $d^+_D(y)\le |Z|+p$ and (2) there are at most $2p+1$ vertices in $Z$ with $d^-_D(z)\le |Y|+p$.
\end{lemma}
\begin{proof}
 We will first prove Part (1). 
	Let $Y_\le$ be the set of vertices in $Y$ with out-degree at most $|Z|+p$ in $D$. Since for every $y\in Y$ and every $z\in Z$ is $yz\in A$,
	 it follows that all vertices in $Y_\le$ have out-degree at most $p$ in $D[Y_\le]$. Hence $\sum_{y\in Y_\le} d^+_{D[Y_\le]}(y)$, \emph{i.e.}, 
	 the sum of out-degrees of vertices in $Y_\le$ in $D[Y_\le]$, is at most $p|Y_\le|$. Hence, 
	 $$\sum_{y\in Y_\le} d^+_{D[Y_\le]}(y)= |A(D[Y_\le])|\le p|Y_\le|.$$
	 Since $D$ is a semicomplete digraph, 
	 $$\frac{|Y_\le|\cdot(|Y_\le|-1)}{2}\le |A(D[Y_\le])|\le  p|Y_\le|.$$ It follows that $|Y_\le|\le 2p+1.$
	 Part (2) follows directly from Part (1) applied to a digraph $D'=(V,A')$ obtained from $D$ by reversing all the arcs i.e. $A'=\{yx\;\mid\; xy\in A \}$.
\end{proof}


Let $D=(V,A)$ be a semicomplete digraph and $t\in [n]$. 
We will call a triple $(Y,Z,S)$ $t$-{\bf valid} if
\begin{enumerate}
	\item $Y,Z$ is a partition of $V(D)$ with $|Y|=t$,
	\item\label{prop:OrderingOfYZ} $S\subseteq V(D)$ is a minimal (w.r.t. inclusion) set such that for all $y\in Y$ and $z\in Z$, if $zy\in A(D)$, then $|\{y,z\}\cap S|\ge 1$,
	\item $|S|\le k$, 
	\item for all $x\in S$, if $d^+_D(x) > n-t+k$, then $x\in Y$, 
	\item for all $x\in S$, if $d^-_D(x) > t+k$, then $x\in Z$. 
\end{enumerate}

We will say a triple $(Y,Z,S)$ is {\bf valid}, if it is $t$-valid for some $t\in \mathbb{N}$.
The following simple observation will help us bound the number of partitions $(Y,Z)$ that could lead to a $t$-valid triple $(Y,Z,S)$. 
 
\begin{lemma}\label{obs:degreePartition}
	For any $t$-valid triple $(Y,Z,S),$ all vertices $v$ with $d^+_{D}(v) > n-t+k$ are in $Y$ and all vertices $v$ with $d^-_{D}(v) > t+k$ are in $Z$. 
\end{lemma}
\begin{proof}
	If $v\in S$, the lemma follows directly from the definition of a $t$-valid triple. 	
	If $v\in V(D)\setminus S$ and $d^+_{D}(v) > n-t+k$, then $v$ has an out-neighbour in $Y\setminus S$, because $|Z\cup S|\le n-t+k$, and $v\in Y$ follows by property~\ref{prop:OrderingOfYZ}. Similarly, if $v\in V(D)\setminus S$ and $d^-_{D}(v) > t+k$, then $v$ has an in-neighbour in $Z\setminus S$ and $v\in Z$ by property~\ref{prop:OrderingOfYZ}.
\end{proof}

\eiben{
\begin{lemma}\label{lemma:sizeOfSetF}
	Let $D=(V,A)$ be a semicomplete digraph, $n=|V|$, and let $t\in [n]$. 
	If there exists a $t$-valid triple, then there are at most $7k+2$ vertices $v$ in $V(D)$ with $d^+_{D}(v) \le n-t+k$ and $d^-_{D}(z) \le t+k$.
\end{lemma}

\begin{proof}
	Let us assume that there is at least one $t$-valid triple and let us denote it $(Y,Z,S)$.
	Note that for all $y\in Y\setminus S$ and $z\in Z\setminus S$ it holds that $zy\notin A(D)$. Since $D$ is a semicomplete digraph, it follows that $yz\in A(D)$. Due to Lemma~\ref{lem:DegreeSeparation} applied to $D-S,$ there are at most $2(k+|Z\cap S|)+1$ 
	vertices in $Y\setminus S$ with $d^+_{D-S}(y)\le |Z\setminus S|+k+|Z\cap S|=n-t+k$
	and there are at most $2(k+|Y\cap S|)+1$ vertices in $Z\setminus S$ with $d^-_{D-S}(z)\le |Y\setminus S|+k+|Y\cap S|=t+k$. 
	Let $F=\{v\in V(D):\ 	d^+_{D}(v)\le n-t+k \mbox{  and } d^-_{D}(v)\le t+k\}.$ By the above, 
	\begin{eqnarray*}
		|F\setminus S| & \le & 2(k+|Z\cap S|)+1 + 2(k+|Y\cap S|)+1\\
		& \le & 4k+2 +\jbj{2|S|}\le 6k+2.
	\end{eqnarray*}
	Thus, $|F|\le 7k+2.$
\end{proof}
} 

\begin{lemma}\label{lemma:numberOfTripples}
	Let $D=(V,A)$ be a semicomplete digraph, $n=|V|$, and let $t\in [n]$. 
	There are at most $2^{8k+2}$ $t$-valid triples $(Y,Z,S)$. 
Moreover, \ee{if we are given the in- and out-degrees of all vertices in $D$ on the input, then} we can enumerate all such triples in time \ee{$\bigoh(2^{8k} k n)$}. 
\end{lemma}
\begin{proof}
	
\eiben{Let $F=\{v\in V(D):\ 	d^+_{D}(v)\le n-t+k \mbox{  and } d^-_{D}(v)\le t+k\}.$ By Lemma~\ref{lemma:sizeOfSetF}, $|F|\le 7k+2$.}
\ee{If the out- and in-degrees of all vertices in $D$ are given on the input,} we can construct the set $F$ in time $\bigoh(n).$
	By Lemma~\ref{obs:degreePartition}, 
	there are at most $2^{7k+2}$ possible partitions $(Y',Z')$ that could lead to a $t$-valid triple $(Y',Z',S')$ for some $S'$, each such partition is uniquely determinate by fixing $Y'\cap F$. 

	For the rest of the proof, we assume that we computed the set $F$ of vertices $v$ in $V(D)$ with $d^+_{D}(v) \le n-t+k$ and $d^-_{D}(v) \le t+k$, $|F|\le 7k+2$. Let $(Y',Z')$ be one of $2^{7k+2}$ partitions that could lead to a $t$-valid triple. 
	 We show that we can enumerate all minimal sets $S'$, $|S'|\le k$, such that for all $y\in Y'$ and $z\in Z'$, if $zy\in A(D)$, then $|\{y,z\}\cap S'|\ge 1$. Let $G$ be an undirected bipartite graph such that $V(G)=V(D)$, the partite sets of $G$ are $Y'$ and $Z',$ and for every $y\in Y'$, $z\in Z',$ it holds $yz\in V(G)$ if and only if $zy\in A(D)$. Then $S'$ is a minimal vertex cover in $G$. Moreover, every minimal vertex cover $S'$ in $G$ leads to a $t$-valid triple $(Y',Z',S')$. It is well known and easy to show that we can enumerate all minimal vertex covers of size at most $k$ in $G$ in time $\bigoh(2^kk^2+kn)$. \ee{This is done by including all vertices with degree at least $k+1$ in every vertex cover. If the resulting graph has more than $k^2$ edges, then there is no vertex of size at most $k$~\cite{BussG93}. Then we can enumerate all vertex covers, by using simple search-tree algorithm that picks an edge, say $uv$, and recursively enumerate all minimal vertex covers of size at most $k-1$ that includes $u$ or $v$, respectively. Given the algorithm, it is also easy to see that} there are at most $2^k$ distinct minimal vertex covers \ee{of size at most $k$.}

It follows that there are at most $2^{7k+2}\cdot 2^k=2^{8k+2}$ $t$-valid triples and we can enumerate all of them in time $\bigoh(n+2^{8k}k^2+kn)=\bigoh(2^{8k}kn)$.
\end{proof}

We are now ready to present our algorithm. 

\begin{theorem} \label{SDfpt}
	There is an FPT algorithm that solves \DCOCk{} on semicomplete digraphs in time \ee{$\bigoh(2^{16k}kn^2)$}.
\end{theorem}
\begin{proof}
	\eiben{Let $D= (V,A)$ be a semicomplete digraph and let $(D,\ell, k)$ be an instance of \DCOCk{}.}
	\subparagraph{Algorithm.} Our algorithm boils down to finding a shortest path in an auxiliary weighted \jbj{acyclic} digraph \jbj{whose vertex set consists} of all the valid triples. \eiben{The main idea is to find a sequence of valid triples $(Y_1,Z_1,S_1),\ldots, (Y_q,Z_q, S_q)$ such that $S=\bigcup_{i\in [q]}S_i$ is a solution for $(D,\ell, k)$ and the strongly connected components of} \viol{$D-X$} \eiben{are subsets of $C_i=Y_{i+1}\setminus (Y_i\cup S)$, where $|C_i|\le \ell$ and for all $i<j$, $x_i\in C_i$, $x_j\in C_j$ it holds that $x_jx_i\notin A$.}
	
	We define \jbj{the} weighted directed \jbj{acyclic} state graph $\mathcal{D}=(\mathcal{V}, \mathcal{A})$ as follows.
	 The set of vertices $\mathcal{V}$ is the set of all $t$-valid triples for all $t\in \{0,1,\dots ,n\}$. 
	 The set of arcs $\mathcal{A}$ contains an arc from a \jbj{$t_1$}-valid triple \jbj{$(Y_1,Z_1, S_1)$} to a \jbj{$t_2$}-valid triple \jbj{$(Y_2,Z_2, S_2)$} if and only if the following conditions holds: 
         \jbj{
         \begin{itemize}
		\item $Y_1\subset Y_2$ (and $Z_2\subseteq Z_1$), 
		\item if $x\in S_1\cap Z_1$ and $x\in Z_2$, then $x\in S_2$,
		\item if $x\in Y_1\setminus S_1$, then $x\in Y_2\setminus S_2$\ee{, and
		\item $|S_1\setminus S_2|+ \max(0, |Z_1\cap Y_2\setminus (S_1\cup S_2)| -\ell) \le k$}.
                \end{itemize}
                }
	We let the weight of an arc from \jbj{$(Y_1,Z_1, S_1)$ to $(Y_2,Z_2, S_2)$} be \jbj{$$|S_1\setminus S_2|+ \max(0, |Z_1\cap Y_2\setminus (S_1\cup S_2)| -\ell).$$} \eiben{This finishes the description of the auxiliary weighted acyclic digraph. In the remainder of the proof we first show that $(D,\ell,k)$ is YES-instance if and only if the cost of the shortest path in $\mathcal{D}$ from $(\emptyset,V(D),\emptyset)$ to $(V(D),\emptyset,\emptyset)$ is at most $k$. 
		Afterwards,
		we bound $|\mathcal{V}|+|\mathcal{A}|$ by $\bigoh(2^{16k}n^2)$ and prove that we can construct the auxiliary digraph in $\bigoh(2^{16k}kn^2)$ time. We can then find a shortest path from $(\emptyset,V(D),\emptyset)$ to $(V(D),\emptyset,\emptyset)$ \jbj{in linear time, that is, in time $\bigoh(2^{16k}n^2)$ since $\mathcal{D}$ is acyclic (by dynamic programming using an acyclic ordering of the vertices), which finishes the proof.}}

\subparagraph*{Correctness of the Algorithm.}
	\jbj{Suppose first that $(D,\ell, k)$ is a} YES-instance of \DCOCk{} such that $D$ is \jbj{a} semicomplete digraph. Let $X$ be a minimum size solution for $(D,\ell, k)$, 
	that is, a minimum size set such that $\mco(D-X)\le \ell$. Since $(D,\ell, k)$ is YES-instance, $|X|\le k$ \jbj{and } $\mco(D-X)\le \ell$,  the vertices of $D-X$ can be partitioned in sets $C_1,\ldots, C_q$ such that 
	\begin{enumerate}
		\item for every $i\in [q]$ is  $|C_i|\le \ell$, and
		\item for every $i,j\in [q]$ with $i<j$ and every $x\in C_i$, $y\in C_j$ we have $xy\in A$ and $yx\notin A$. 
	\end{enumerate}

Our goal is to define a sequence of valid triples $(Y_i,Z_i,S_i)$, $i\in [q]$, such that the arc $((Y_i,Z_i,S_i),(Y_{i+1},Z_{i+1},S_{i+1}))$ is in $\mathcal{A}$ and the cost of the path in $\mathcal{D}$ defined by this sequence is $|X|$. 
We will \eiben{construct} these triples \eiben{from $X$ and $C_1,\ldots, C_q$ with some additional restrictions that makes it easier to show that they indeed define a path in $\mathcal{D}$ of cost at most $|X|$. Namely, we will define them} such that for all $i,j\in [q]$, $i<j$ \eiben{the triples satisfy the following properties:}
\begin{enumerate}
	\item $(Y_i,Z_i,S_i)$ is $t_i$-valid for some $t_i\in [n]$, 
	\item  $C_1\cup\cdots \cup C_i\subseteq Y_i$, 
	\item $C_{i+1}\cup\cdots \cup C_q\subseteq Z_i$, 
	\item $S_i\subseteq X$,
	\item $Y_i\subset Y_j$ and $Z_j\subseteq Z_i$, 
	\item if $x\in S_i\cap Z_i$ and $x\in Z_j$, then $x\in S_j$,
	\item if $x\in Y_i\setminus S_i$, then $x\in Y_j\setminus S_j$.
\end{enumerate} 

It is straightforward to verify that, given the above properties, the arc $$((Y_i,Z_i,S_i),(Y_{i+1},Z_{i+1},S_{i+1}))\in \mathcal{A}.$$  
We first show that a sequence with the above properties indeed exists and defer the computation of the cost of the path defined by this sequence to later.

To obtain this sequence, we need to discuss how to distribute \jbj{the vertices of }$X$ in \jbj{the sets} $Y_i$'s and $Z_i$'s and how to compute $S_i$, $S_j$ \eiben{(Note that the partition of the vertices in $V\setminus X$ is fixed by properties 2 and 3)}. 

To distribute \jbj{the vertices of} $X$ between $Y_i$ and $Z_i$, we put all $x\in X$ with $d^+_D(x) \ge n-t_i+k$ \jbj{in} $Y_i$ and all $x\in X$ with $d^-_D(x) \ge t_i+k$ \jbj{in}  $Z_i$. The remaining vertices in $X$ we can distribute arbitrarily, we only have to make sure that for all $i,j\in [q]$, $i<j$, it holds that $Y_i\subset Y_j$ and $Z_j\subseteq Z_i$. Now $|X|\le k$ and for all $y\in Y_i\setminus X = C_1\cup\cdots \cup C_i$ and $z\in Z_i\setminus X= C_{i+1}\cup\cdots \cup C_q$ we have $zy\notin A(D)$. The set $S_i$ is defined \jbj{to be those} vertices $x\in X$ such that one of the following holds: 
\begin{enumerate}
	\item $x\in Y_i$ and there exists $z\in Z_i\setminus X$ such that $zx\in A(D)$, 
	\item $x\in Z_i$ and there is an arc $xy\in A(D)$, $y\in Y_i$ such that $y\notin S_i$.
\end{enumerate}

Note that all arcs from $Z_i$ to $Y_i$ are covered by $S_i$ and for each $x\in X$ there is an  arc $zy$ \jbj{from $Z_i$ to $Y_i$} with $\{y,z\}\cap X= \{x\}$. 
Note that if $x\in Y_i\setminus S_i$, then $x\in Y_j\setminus S_j$ for all $j>i$. On the other hand, if $x\in Z_i\cap S_i$, then there is $y\in Y_i\setminus S_i$ such that $xy\in A(D)$. Moreover, for all $j>i$, $y\in Y_j\setminus S_j$. Therefore, if $x\in Z_j$, then $x\in S_j$. From the above two properties it follows that if $x\in S_i\setminus S_j$, then $x\notin S_{j+1}\cup\cdots\cup S_q$. This finishes the \eiben{proof of the existence of a sequence of valid triples $(Y_1, Z_1, S_1),\ldots, (Y_q, Z_q, S_q)$ with properties 1-7.}

We claim that the cost of path following this sequence is \ee{$|X|\le k$}.
First note that if $x\in S_i\setminus S_{i+1}$, then $x\in Y_{i+1}$ and for all $j\ge i+1$ it holds $x\notin S_j$, hence every vertex in $X$ is counted in at most one of the sets $S_{i}\setminus S_{i+1}$. Now the set $C_i$ is precisely $(Z_{i-1}\cap Y_i)\setminus X$. If $x\in Z_{i-1}\cap Y_i\cap X$ is in some set $S_j$, then
\eiben{from the properties 5,  6 and 7 of the sequence of triples it follows that}
 $x$ is in $S_{i-1}\cup S_i$. Hence $|Z_{i-1}\cap Y_i\setminus (S_{i-1}\cup S_i)| -|C_i|)$ is precisely the number of vertices in $X$ that are in $Z_{i-1}\cap Y_i$ and in none of the sets $S_j$, $j\in [q]$. Note that for such vertex \eiben{$x\in (Z_{i-1}\cap Y_i)\setminus\bigcup_{j\in [q]} S_j$} and \eiben{a vertex} $y\in Y_{j}\setminus S_j$, \eiben{for some $j\in [q]$ with $j<i$}, \eiben{it holds} $xy\notin A(D)$ \eiben{(else by definition of a valid triple $|\{x,y\}\cap S_j|\ge 1$)}. Similarly for $z\in Z_j\setminus S_j$, $j>i$, $zx\notin A(D)$. Hence, if $|C_i|<\ell$, then $X\setminus\{x\}$ \jbj{would be} a smaller solution \jbj{for the instance} $(D,\ell, k)$ and because of minimality of $X$,  $(|Z_{i-1}\cap Y_i\setminus (S_{i-1}\cup S_i)| -\ell)$ is precisely the number of vertices in $X$ that are in $Z_{i-1}\cap Y_i$ and in none of the sets $S_j$. It follows that each vertex in $X$ is counted on precisely one arc on the path and the shortest path
\jbj{from $(\emptyset{},V(D),\emptyset)$ to $(V(D),\emptyset{},\emptyset{})$} in $\mathcal{D}=(\mathcal{V}, \mathcal{A})$ has length precisely $|X|$.




For the other direction, let \jbj{some} shortest path in $\mathcal{D}$ from $(\emptyset,V(D),\emptyset)$ to $(V(D),\emptyset,\emptyset)$ be defined by the sequence $(Y_i,Z_i,S_i)$, $i\in \{0,\ldots,q\}$, and assume that the cost of the path is at most $k$. 
\eiben{For every $i\in [q]$, let $T_i$}
be \jbj{an arbitrary set consisting of } $(|(Z_{i-1}\cap Y_i)\setminus (S_{i-1}\cup S_i)| -\ell)$ vertices \jbj{from}  $Z_{i-1}\cap Y_i\setminus (S_{i-1}\cup S_i)$ and  let $X=\bigcup_{i\in [q]}(T_i\cup S_i)$.  
Because  the pair $((Y_{i-1},Z_{i-1},S_{i-1}),(Y_i,Z_i,S_i))$ is an arc in $\mathcal{D}$ \jbj{for every $i\in [q]$}, we have $Y_{i-1}\subseteq Y_i$ and $Z_i\subseteq Z_{i-1}$. Moreover, $(Y_{i-1},Z_{i-1},S_{i-1})$ and $(Y_i,Z_i,S_i))$ are $t_{i-1}$-valid and $t_i$-valid triples, for some $t_{i-1},t_i\in [n]$, respectively. Therefore, there is no arc from $Z_j\setminus X$ to $Y_i\setminus X$ for any $i\le j\in [q]$. It follows that each strongly connected component \eiben{of $D-X$} is a subset of \eiben{$(Z_{i-1}\cap Y_i)\setminus X$ for some $i\in [q]$. In particular note that $(Z_{i-1}\cap Y_i)\cap X = (Z_{i-1}\cap Y_i)\cap (S_{i-1}\cup S_i\cup T_i)$, $(S_{i-1}\cup S_i)\cap T_i=\emptyset$ and $T_i\subseteq (Z_{i-1}\cap Y_i)$. Hence the size of each connected component is at most $\max_{i\in [q]}|Z_{i-1}\cap Y_i)\setminus (S_{i-1}\cup S_i\cup T_i)| = \max_{i\in [q]}|\left((Z_{i-1}\cap Y_i)\setminus (S_{i-1}\cup S_i)\right)\setminus T_i| = \max_{i\in [q]}(|\left((Z_{i-1}\cap Y_i)\setminus (S_{i-1}\cup S_i)\right)| - |T_i|)\le \ell$. }
Since $S_0=S_q=\emptyset$, every vertex that appears in \eiben{$S_i$} for some $i\in [q]$ is counted in \eiben{some $|S_j\setminus S_{j+1}|$, where $j\ge i$ and every vertex that appears in $T_i$ for some $i\in [q]$ is counted in $\max(0, |Z_i\cap Y_{i+1}\setminus (S_i\cup S_{i+1})| -\ell)$ }
and the final set $X$ has at most $k$ vertices. 

\subparagraph*{Construction of the Auxiliary Weighted Digraph.}
Note that by Lemma~\ref{lemma:numberOfTripples}, $|\mathcal{V}|\le 2^{8k+2}n$ and\ee{, since we can compute the out- and in-degrees of all vertices in $D$ in time  $\bigoh(n^2)$}, we can enumerate all vertices \ee{in $\mathcal{D}$} in time $\bigoh(2^{8k}kn^2)$. \eiben{It follows that $|\mathcal{A}|\le |\mathcal{V}|^2\le 2^{16k+4}n^2$ and $|\mathcal{V}|+|\mathcal{A}|=\bigoh(2^{16k}n^2)$. It remains to show that} for a pair of triples \jbj{$(Y_1,Z_1, S_1)$ and $(Y_2,Z_2, S_2)$,} we can check whether \jbj{$((Y_1,Z_1, S_1),(Y_2,Z_2, S_2))$ }is an arc and compute its weight \ee{in $\bigoh(k)$ amortized time. First note that if $|Y_1|\ge |Y_2|$, then the arc is not there. We will only check if $((Y_1,Z_1, S_1),(Y_2,Z_2, S_2))$ is an arc if $|Y_1|<|Y_2|$. This can be done without computing the sizes of $Y_1$ and $Y_2$, respectively, if we enumerate the $t$-valid triples in $\mathcal{D}$ in levels in the order increasing $t$ \eiben{(i.e., we invoke Lemma~\ref{lemma:numberOfTripples} for $t$ only after we added all $t'$-valid triples, for all $t'<t$, to $\mathcal{V}$.)} and compute all in-neighbours of a vertex when it is added to $\mathcal{V}$. Moreover, when adding the triple $(Y,Z,S)$ in $\mathcal{V}$, we will in $\bigoh(n)$ time compute maps $\alpha_{(Y,Z,S)}: V(D)\rightarrow \{0,1\}$ such that $\alpha_{(Y,Z,S)}(x)=0$ if and only if $x\in Y$ and $\beta_{(Y,Z,S)}: V(D)\rightarrow \{0,1\}$ such that $\beta_{(Y,Z,S)}(x)=0$ if and only if $x\in S$. We also compute the set $\jbj{\Delta_{Y,Z}} = \{x\mid x\in V(D),  d^+_D(x)\le |Z|+k, d^-_D(x)\le |Y|+k\}$.} \eiben{By Lemma~\ref{lemma:sizeOfSetF}, $|\Delta_{Y,Z}|\le 7k+2$.} \ee{ Now we can describe the $\bigoh(k)$ algorithm that determines whether $((Y_1,Z_1, S_1),(Y_2,Z_2, S_2))$ is an arc. 
	
	First, for every $x\in S_1$ we can in constant time check that $x\in S_1\cap Z_1$ (\emph{i.e.,} $\alpha_{(Y_1,Z_1,S_1)}(x)=1$ and  $\beta_{(Y_1,Z_1,S_1)}(x)=0$) and $x\in Z_2$ ($\alpha_{(Y_2,Z_2,S_2)}(x)=1$ ) implies $x\in S_2$ ($\beta_{(Y_2,Z_2,S_2)}=0$). Similarly we can check in constant time that if $x\in Y_1$, then $x\in Y_2\setminus S_2$.
	
	Second, by Lemma~\ref{obs:degreePartition} and since $|Y_1|<|Y_2|$ and $|Z_1|>|Z_2|$, we get that to check that $Y_1\subset Y_2$ and $Z_2\subseteq Z_1$, we only need to check for every $x\in \jbj{\Delta_{Y_1,Z_1}\cup \Delta_{Y_2,Z_2}}$ that $\alpha_{(Y_1,Z_1,S_1)}(y)=0$ implies $\alpha_{(Y_2,Z_2,S_2)}\jbj{(x)}=0$. This check can done in $\bigoh(\jbj{|\Delta_{Y_1,Z_1}\cup \Delta_{Y_2,Z_2}|}) = \bigoh(k)$ time. 
	Finally, to compute the weight of the arc, we note that $|Z_1\cap Y_2|$ is precisely $|Y_2|-|Y_1|$, because $Y_1\subset Y_2$ and $Z_1=V(D)\setminus Y_1$, so we only need to check how many of vertices in $S_1\cup S_2$ are in $Z_1\cap Y_2$ and how many of vertices in $S_1$ are also in $S_2$. Moreover, we only need to compute $|Z_1\cap Y_2\setminus (S_1\cup S_2)| -\ell)$ if $\ell < |Y_2|-|Y_1| \le \ell + 2k$. Else either the weight of the arc is precisely $|S_1\setminus S_2|$ or it would be more than $k$ and hence it is not an arc. Hence, we end up spending $\bigoh (k+\log n)$ time on computation of weight of each of at most $\bigoh(2^{16k}kn)$ many arcs (for which $\ell < |Y_2|-|Y_1| \le \ell + 2k$ ) and $\bigoh(k)$ on all of at most $\bigoh(2^{16k}n^2)$ remaining arcs. Since $k\le n$, }
we can construct $\mathcal{D}$ in \ee{$\bigoh(2^{16k}kn^2)$} time.
\end{proof}


\ee{In the rest of the section, we will show that the dependency on both $k$ and $n$ cannot be significantly improved. More precisely, we will show an unconditional lower-bound of $\Omega(n^2)$ even if $k=0$, as we show that we need to read at least $\Omega(n^2)$ arcs of the input instance in the worst case to distinguish between $k=0$ and $k=1$. Furthermore, we show that any $2^{o(k)}n^{\bigoh(1)}$ algorithm would imply that Exponential Time Hypothesis fails. 
	\begin{theorem}\label{thm:n^2}
		There is no deterministic sequential algorithm 
		that outputs \jbj{the} correct answer for every instance $(D,\ell,0)$ of \DCOC{} \jbj{when} $D$ is a tournament in $o(n^2)$ time.
	\end{theorem}
\begin{proof}
		For $i\in \mathbb{N}$, let $H_i$ be an arbitrary but fixed strongly connected tournament on $i$ vertices. 
		If $\frac{n}{2} \le \ell < n$, then let us consider the graph $D$ obtained by taking
		disjoint union of $H_{\lfloor\frac{n}{2}\rfloor}$ and $H_{\lceil\frac{n}{2}\rceil}$ and orienting arcs between $H_{\lfloor\frac{n}{2}\rfloor}$ and $H_{\lceil\frac{n}{2}\rceil}$ from $H_{\lfloor\frac{n}{2}\rfloor}$ to $H_{\lceil\frac{n}{2}\rceil}$. Clearly, $\mco(D)={\lceil\frac{n}{2}\rceil}\le \ell$ and $(D,\ell, 0)$ is YES-instance of \DCOC{}. Note there are ${\lfloor\frac{n}{2}\rfloor}\cdot{\lceil\frac{n}{2}\rceil} = \Theta(n^2)$ arcs between $H_{\lfloor\frac{n}{2}\rfloor}$ and $H_{\lceil\frac{n}{2}\rceil}$. Now let $\mathbb{A}$ be a deterministic sequential algorithm that solves \DCOCk{} in $o(n^2)$ time if $k=0$. If we run $\mathbb{A}$ on $D$, then there is an arc from $H_{\lfloor\frac{n}{2}\rfloor}$ to $H_{\lceil\frac{n}{2}\rceil}$ that $\mathbb{A}$ did not read. Let this arc be $xy$ and let $D'$ be a graph obtained from $D$ by replacing the arc $xy$ by the arc $yx$. It follows that $D'$ is strongly connected and hence $(D',\ell, 0)$ is NO-instance of \DCOC{}. However, because the algorithm $\mathbb{A}$ decided that $(D,\ell,0)$ is YES-instance without considering the orientation of the arc between $x$ and $y$ on the instance $(D,\ell, 0)$ and the only difference between $(D,\ell, 0)$ and $(D',\ell, 0)$ is the orientation of the arc between $x$ and $y$, it follows that $\mathbb{A}$ outputs that $(D',\ell, 0)$ is YES-instance, which contradicts the assumption that $\mathbb{A}$ outputs correct answer for every instance $(D,\ell,0)$ of \DCOC{} such that $D$ is a tournament. 
		
		If $\ell < \frac{n}{2}$, the proof is very similar to the above, the only difference is the construction of the digraph $D$. To construct $D$ we first take the disjoint union of $q= \lfloor\frac{n}{\ell}\rfloor$ copies of $H_\ell$, \jbj{denoted} $H^1_\ell, \ldots, H^{q}_\ell$, and \jbj{one copy of} $H_{n-q\ell}$. We add the arc $xy$ to $D$ if $x\in H^i_\ell$ and $y\in H^j_\ell$ such that $1\le i<j \le q$ or if $x\in H^i_\ell$, \jbj{$i\in [q]$}, and $y\in H_{n-q\ell}$. It follows that $D$ is a tournament and $\mco(D)=\ell$, that is $(D,\ell, 0)$ is a YES-instance. Now let $Y=\bigcup_{i\in [\lfloor\frac{q}{2}\rfloor]} V( H^{i}_\ell)$ and $Z=V(D)\setminus Y$. It is easy to see that $\frac{n}{4}\le |Y|\le \frac{n}{2}$ and there are $\Theta(n^2)$ arcs from $Y$ to $Z$ in $D$. Moreover if $yz\in A(D)$ is an arc such that $y\in Y$ and $z\in Z$, then $D_{yz}=(V(D), (A(D)\setminus \{yz\})\cup \{zy\})$ contains a strongly connected components of size at least $2\ell$ that includes all vertices in $V(H^{\lfloor\frac{q}{2}\rfloor}_{\ell})\cup V(H^{\lfloor\frac{q}{2}\rfloor+1}_{\ell})$. The proof follows by analogous arguments to the case $n-\ell < \ell$, as for any algorithm $\mathbb{A}$ that solves $(D,\ell, k)$ in $o(n^2)$, there is an arc $yz$ such that $\mathbb{A}$ outputs incorrectly that $(D_{yz},\ell, k)$ is YES-instance. 
\end{proof}

Finally, we will present our \gzg{$\bigoh^*(2^{o(k)})$} lower bound result, based on the well-established Exponential Time Hypothesis (ETH).


Our result uses the fact that the classical \textsc{Vertex Cover} problem cannot be solved in subexponential time under ETH.	

\begin{theorem}[Cai and Juedes \cite{CaiJuedes03}]\label{thm:ethVC}
	There is no $2^{o(k)}\cdot |V(G)|^{\mathcal{O}(1)}$ algorithm for \textsc{Vertex Cover}, unless ETH fails.
\end{theorem}

Given the above result by Cai and Juedes, the lower bound then directly follows from the proof of NP-harness of {\sc Directed Feedback Vertex Set} by Speckenmeyer \cite{Speckenmeyer89}. In fact, given a graph $G$, Speckenmeyer constructs in $O(|V(G)|^2)$ time a tournament $T$ with $3|V(G)|-2$ vertices such that for every $k$ the graph $G$ has a vertex cover of size at most $k$ if and only if $T$ has a directed feedback vertex set of size at most $k$ (see Theorem~6 in \cite{Speckenmeyer89}). Hence, we obtain the following:

\begin{theorem}\label{thm:o(k)}
	There is no algorithm solving \DCOCk{} on tournaments in time $2^{o(k)}n^{\bigoh(1)}$, unless ETH fails. 
\end{theorem}


}

In Theorem~\ref{SDfpt} we saw that there is an FPT algorithm for \DCOCnell{} that runs in  $\bigoh^*(2^{16(n-\ell)})$ time,
as we may assume that $k \leq n - \ell$.
\gutin{
By the construction explained before Theorem \ref{thm:o(k)} we can replace $k$ by $n$ in $2^{o(k)}$ in Theorem \ref{thm:o(k)} and thus 
obtain a matching lower bound for the upper bound $\bigoh^*(2^{16(n-\ell)}).$

\begin{theorem}\label{thm:lb4n-ell}
   There is no $2^{o(n-\ell)} n^{\bigoh(1)}$-time algorithm for solving \DCOCnell{} on semicomplete digraphs, unless ETH fails.

\end{theorem}
}

\section{Conclusions} \label{sec:c}

Since  \DCOC{} generalizes  {\sc Directed Feedback Vertex Set}, it would likely be hard to improve our upper bound and obtain a tight lower bound for the time complexity of \DCOCellk{} on general digraphs. 
It seems easier to improve our upper and lower bounds on the time complexity of \DCOCk{} on semicomplete digraphs.  
There are several digraph classes which generalize semicomplete digraphs such as semicomplete multipartite digraphs and quasi-transitive digraphs, see e.g.  \cite{bang2018}. It'd be interesting to consider the time complexity of the problem on such digraphs.



\begin{thebibliography}{10}

\bibitem{bang2009}
J.~Bang-Jensen and G.~Gutin.
\newblock {\em {Digraphs: Theory, Algorithms and Applications}}.
\newblock Springer-Verlag, London, 2nd edition, 2009.

\bibitem{bang2018}
J.~Bang{-}Jensen and G.Z. Gutin, editors.
\newblock {\em Classes of Directed Graphs}.
\newblock Springer Monographs in Mathematics. Springer, 2018.

\bibitem{Bang-JensenT92}
J.~Bang{-}Jensen and C.~Thomassen.
\newblock A polynomial algorithm for the 2-path problem for semicomplete
  digraphs.
\newblock {\em {SIAM} J. Discrete Math.}, 5(3):366--376, 1992.
\newblock {http://dx.doi.org/10.1137/0405027}
  {\path{doi:10.1137/0405027}}.

\bibitem{BshoutyG17CIAC}
N.H. Bshouty and A.~Gabizon.
\newblock Almost optimal cover-free families.
\newblock In Dimitris Fotakis, Aris Pagourtzis, and Vangelis~Th. Paschos,
  editors, {\em Algorithms and Complexity - 10th International Conference,
  {CIAC} 2017, Athens, Greece, May 24-26, 2017, Proceedings}, volume 10236 of
  {\em Lecture Notes in Computer Science}, pages 140--151, 2017.

\bibitem{BussG93}
Jonathan~F. Buss and Judy Goldsmith.
\newblock Nondeterminism within {P}.
\newblock {\em {SIAM} J. Comput.}, 22(3):560--572, 1993.

\bibitem{CaiJuedes03}
Liming Cai and David Juedes.
\newblock On the existence of subexponential parameterized algorithms.
\newblock {\em J. Comput. System Sci.}, 67(4):789--807, 2003.

\bibitem{chenJACM55}
J.~Chen, Y.~Liu, S.~Lu, B.~O'Sullivan, and I.~Razgon.
\newblock A fixed-parameter algorithm for the directed feedback vertex set
  problem.
\newblock {\em J. Assoc. Comput. Mach.}, 55(5):21:1--21:19, 2008.

\bibitem{cygan2015}
M.~Cygan, F.V. Fomin, L.~Kowalik, D.~Lokshtanov, D.~Marx, M.~Pilipczuk,
  M.~Pilipczuk, and S.~Saurabh.
\newblock {\em Parameterized Algorithms}.
\newblock Springer, 2015.

\bibitem{DowneyF99}
R.G. Downey and M.R. Fellows.
\newblock {\em Parameterized Complexity}.
\newblock Monographs in Computer Science. Springer, 1999.
\newblock  {http://dx.doi.org/10.1007/978-1-4612-0515-9}
  {\path{doi:10.1007/978-1-4612-0515-9}}.

\bibitem{downey2013}
R.G. Downey and M.R. Fellows.
\newblock {\em Fundamentals of Parameterized Complexity}.
\newblock Springer, 2013.

\bibitem{DrangeDH16}
P.G. Drange, M.S. Dregi, and P.~van~'t Hof.
\newblock On the computational complexity of vertex integrity and component
  order connectivity.
\newblock {\em Algorithmica}, 76(4):1181--1202, 2016.

\bibitem{GokeMM20arXiv}
A.~G{\"{o}}ke, D.~Marx, and M.~Mnich.
\newblock Parameterized algorithms for generalizations of directed feedback
  vertex set.
\newblock {\em CoRR}, abs/2003.02483, 2020.
\newblock Preliminary version published in Proc. CIAC 2019.
\newblock {http://arxiv.org/abs/2003.02483} .

\bibitem{gross13}
D.~Gross, M.~Heinig, L.~Iswara, L.~W. Kazmierczak, K.~Luttrell, J.~T. Saccoman,
  and C.~Suffel.
\newblock A survey of component order connectivity models of graph theoretic
  networks.
\newblock {\em WSEAS Transactions on Mathematics}, 12:895--910, 2013.

\bibitem{ImpagliazzoRF2001}
Russell Impagliazzo, Ramamohan Paturi, and Francis Zane.
\newblock Which problems have strongly exponential complexity?
\newblock {\em J. Comput. System Sci.}, 63(4):512--530, 2001.
\newblock Special issue on FOCS 98 (Palo Alto, CA).

\bibitem{KumarL16}
M.~Kumar and D.~Lokshtanov.
\newblock A $2lk$ kernel for $l$-component order connectivity.
\newblock In J.~Guo and D.~Hermelin, editors, {\em 11th International Symposium
  on Parameterized and Exact Computation, {IPEC} 2016, August 24-26, 2016,
  Aarhus, Denmark}, volume~63 of {\em LIPIcs}, pages 20:1--20:14. Schloss
  Dagstuhl - Leibniz-Zentrum fuer Informatik, 2016.

\bibitem{Lee19}
E.~Lee.
\newblock Partitioning a graph into small pieces with applications to path
  transversal.
\newblock {\em Math. Program.}, 177(1-2):1--19, 2019.

\bibitem{Neogi+}
Rian Neogi, M.~S. Ramanujan, Saket Saurabh, and Roohani Sharma.
\newblock On the parameterized complexity of deletion to $\mathcal{H}$-free
  strong components, 2020.
\newblock To appear in Proceedings of MFCS 2020.
\newblock  {http://arxiv.org/abs/2005.01359} {\path{arXiv:2005.01359}}.

\bibitem{Speckenmeyer89}
E.~Speckenmeyer.
\newblock On feedback problems in digraphs.
\newblock In M.~Nagl, editor, {\em Graph-Theoretic Concepts in Computer
  Science, 15th International Workshop, {WG} '89, 1989, Proceedings}, volume
  411 of {\em Lecture Notes in Computer Science}, pages 218--231. Springer,
  1989.

\end{thebibliography}
\end{document}